\documentclass[a4paper]{article}
\usepackage{algorithm}
\usepackage{algorithmic}
\usepackage{graphicx} 
\usepackage{natbib}
\usepackage{fullpage}
\usepackage{authblk}
\usepackage{epstopdf}

\newcommand{\comment}[1]{}

\newtheorem{theorem}{Theorem}
\newtheorem{lemma}{Lemma}

\newtheorem{claim}{Claim}
\newtheorem{corollary}{Corollary}

\newenvironment{proof}{{\bf Proof:}}{\hfill\rule{1.5mm}{3mm}\vspace{0.1in}}

\newcommand{\qed}{}
\begin{document}
\bibliographystyle{abbrv}

\title{A simpler load-balancing algorithm for range-partitioned data
  in Peer-to-Peer systems}
\author{Jakarin Chawachat\thanks{Corresponding author}}
\author{Jittat Fakcharoenphol}
\affil{Department of Computer Engineering\\
Kasetsart University\\
Bangkok Thailand
}

\maketitle
\let\thefootnote\relax\footnote{ This work has been partially supported by KU-RDI grant number Wor-Tor(Dor) 84.53.\\
\textbf{Email address:} cjakarin@gmail.com (Jakarin Chawachat), jittat@gmail.com (Jittat Fakcharoenphol)}
\begin{abstract}
  Random hashing is a standard method to balance loads among nodes in
  Peer-to-Peer networks. However, hashing destroys locality properties of
  object keys, the critical properties to many applications, more
  specifically, those that require range searching. To preserve a key
  order while keeping loads balanced, Ganesan, Bawa and Garcia-Molina
  proposed a load-balancing algorithm that supports both object
  insertion and deletion that guarantees a ratio of 4.237 between the
  maximum and minimum loads among nodes in the network using constant
  amortized costs. However, their algorithm is not straightforward to
  implement in real networks because it is recursive.  Their algorithm
  mostly uses local operations with global max-min load information.
  In this work, we present a simple non-recursive algorithm using
  essentially the same primitive operations as in Ganesan {\em et
    al.}'s work.  We prove that for insertion and deletion, our
  algorithm guarantees a constant max-min load ratio of 7.464 with
  constant amortized costs.
\end{abstract}

\newcommand{\minbalance}{{\sc MinBalance}}
\newcommand{\split}{{\sc Split}}
\newcommand{\splitmax}{{\sc SplitMax}}
\newcommand{\splitnbr}{{\sc SplitNbr}}
\newcommand{\nbradj}{{\sc NbrAdjust}}
\newcommand{\reorder}{{\sc Reorder}}
\newcommand{\adjload}{{\sc AdjustLoad}}

\section{Introduction}

One of important issues in Peer-to-Peer (P2P) networks is load balancing.
Load balancing is a method that balances loads among nodes in the
networks. In P2P networks, a standard technique usually used
to spread keys over nodes is hashing.  The research on the
construction of distributed hash tables (DHTs) is very active
recently. However, hashing destroys locality properties of keys. This
makes many applications difficult to support a range searching.

Ganesan, Bawa and Garcia-Molina~\cite{GanesanBGM04-vldb} proposed a
sophisticated load-balancing algorithm on top of linearly ordered
buckets. Since the ordering of keys is preserved, item searching is
simple and  range searching is naturally supported.  Their algorithm,
called the {\adjload} algorithm, uses two basic balancing operations,
{\nbradj} and {\reorder} with global max-min information, and can
maintain a good ratio of 4.237 between the maximum and minimum loads
among nodes in the network, while requiring a constant amortized work
per operation. We give the description of their algorithm in
Section~\ref{sect:GBGM}.

Although both operations are easy to state, the {\adjload} algorithm
is recursive; thus, it is not straightforward to implement in
distributed environments.  Also, in the worst case, the {\adjload}
algorithm does not guarantee the number of invoked balancing
operations, the number of updated data (partition change and load
information) and the number of affected nodes.  Since each balancing
operation may require global information, the number of global queries
may be higher than the number of insertions and deletions.

In this paper, we present a simpler, non-recursive load-balancing
algorithm that uses the same primitive operations, {\nbradj} and
{\reorder} as in Ganesan {\em et al.}, but for each insertion or
deletion, primitive operations are called at most once.  
This also implies that our proposed algorithm 
only makes queries to global at most once per insertion or deletion.  

As in Ganesan {\em et al.}, we prove that the ratio between the maximum
load and the minimum load, the imbalance ratio, is at most 7.464 and
the amortized cost of the algorithm is a constant per operation.

Our algorithm uses two high-level operations, {\minbalance} and
{\split}:  

\begin{enumerate}

\item {\bf The {\minbalance} operation} occurs when there is an
  insertion at some node $u$ causing the load of $u$ to be too high;
  in this case, we shall take the node $v$ with the minimum load, transfer
  its load to one of $v$'s neighbors with a lighter load, and let $v$
  share half the load with $u$.

\item {\bf The {\split} operation} occurs when there is a deletion at
  some node $u$ causing the load of $u$ to be too low; in this case,
  we either let $u$ take some load from one of its neighbors, or
  transfer all $u$'s load to its neighbors and let it share half the
  load with the maximum loaded node.

\end{enumerate}

\textbf{Overview of the techniques.}  When considering only insertion,
the key to proving the imbalance ratio is to analyze how the load of 
neighbors of the minimum loaded node changes over time.  Let $z$
denote the lightly-loaded neighbor of $v$.  The bad situation can occur
when $z$ takes  entire loads of the minimum loaded nodes (at
various points in time) for too many times; this could cause the load
of $z$ to be too high compared with the minimum load.  We show that this
is not possible because  each time a load is transferred to $z$, the
minimum load is increased to within some constant factor of the load
of $z$, thus keeping the imbalance ratio bounded.

We make an important assumption in our analysis of the insertion-only
case, namely, that the minimum load can never decrease.  For the
general case, this is not true.  We maintain the general analysis
framework by introducing the notion of phases such that within a
phase the minimum load remains monotonically non-decreasing.  To prove
the result in this case, we show that inside a phase, the imbalance
ratio is small; and when we enter a new phase (i.e., when the minimum
load decreases), we are in a good starting condition.

\textbf{Organization.} The remainder of the paper is organized
as follows. In Section~\ref{sec:relatedwork}, we discuss related
work. Section~\ref{sect:model} describes the model, states the basic
definitions and reviews the {\adjload} algorithm. In
Section~\ref{sect:algorithm-insertonly}, we consider the insert-only
case. We propose the {\minbalance} operation, analyze the imbalance
ratio and calculate the cost of the algorithm. In
Section~\ref{sect:algorithm-delete}, we consider the general case,
which has both insertion and deletion.  We propose the {\split}
operation that is used when a deletion occurs and analyze the imbalance
ratio and the cost of the load-balancing algorithm.

\section{Related work}
\label{sec:relatedwork}

Research on complex queries in P2P networks has long been an
interesting problem. It started when Harren {\em et
  al.}~\cite{Harren02complexqueries} argued that  complex queries
are important open issues in P2P networks. After that, much research has appeared on search methods in P2P networks. For more
information, readers are referred to the survey on searching in P2P
networks by Risson and Moors~\cite{Risson06Survey}.

Range searching is one of the search methods that arises in many fields
including P2P systems. Since most data structures for distributed
items are based on distributed hash tables (DHTs)~\cite{Kademlia02, CAN01,
  Pastry01,  Chord03, tapestry04}, early work focused mainly
on building range search data structures on top of DHTs, e.g.,
PHT~\cite{PHT2004} and DST~\cite{DST06}, based on binary search trees
and segment trees, respectively. These data structures do not address a
load-balancing mechanism when supporting insertion and deletion. Moreover,
hot spots may occur when loads are highly imbalanced.

There are other data structures for range searching in P2P networks,
which are not based on DHTs. SkipNet~\cite{SkipNet03}, which is adapted
from Skip Lists~\cite{skiplist90}, can support range searching or load
balancing but not both. Skip Graphs~\cite{Aspnes-SODA03} is also adapted from Skip Lists and
addresses on the range searching but it does not address on load balancing on the number of items per nodes.

Many data structures support efficient range queries and show a good 
load balancing property in experiments, e.g.,
Mercury~\cite{Mercury_sigcomm04}, Baton~\cite{baton05},
 Chordal graph~\cite{Joung08}, Dak~\cite{Dak06} and
Yarqs~\cite{Yarqs09}. However, they do not have any
theoretical guarantees on load distribution among nodes in their data
structures.

There are mainly two groups of researchers trying to address both range searching and
load-balancing theoretically. The former is the group of Karger and Ruhl~\cite{Karger03newalgorithms, 
Karger04simpleefficient}. The latter is the group of Ganesan and Bawa~\cite{Ganesan03} 
and Ganesan, Bawa and Garcia-Molina~\cite{GanesanBGM04-vldb}. 
Both of them use two operations. The first operation balances loads between two consecutive nodes 
by transferring the load from the node with higher load to the node with lower load. For the second operation, a node $i$ 
transfers its entire load to its neighbor and relocates its position to share load with 
a node $j$ in a new position.

Karger and Ruhl~\cite{Karger03newalgorithms, Karger04simpleefficient} 
presented the randomized protocol, where each node chooses another node 
to perform a balancing operation at random.  Load balancing operations
should be performed regularly, even when there is no insertion or deletion
at the node.  More precisely, they showed that if each node contacts $\Omega(\log n)$ 
other nodes then the load of each node is between $\frac{\epsilon}{16}\cdot L$ and 
$\frac{16}{\epsilon}\cdot L$, where $L$ denotes the average load and $\epsilon$ is a 
constant with $0 < \epsilon < \frac{1}{4}$, with high probability, where the hidden constant
in the $\Omega$ notation depends on $\epsilon$.  They also showed that
the cost of load balancing steps can be amortized over the constant costs of 
insertion and deletion.  Again, the constants depend on the value of $\epsilon$.
We note that this implies a high probability bound of at least 4,096 for the imbalance ratio.

Ganesan and Bawa~\cite{Ganesan03} and Ganesan, Bawa and
Garcia-Molina~\cite{GanesanBGM04-vldb} proposed a distributed
load-balancing algorithm that works on top of any linear data
structures of items. Their algorithm is recursive and uses 
information about the maximum and minimum-loaded nodes.  They guarantee a
constant imbalance ratio with a constant cost per insertion and
deletion.  The ratio can be adjusted to $4.237$.  This is much smaller
than that of Karger and Ruhl.  The major drawback is that their
algorithm requires  global knowledge of the maximum and 
minimum-loaded nodes.  
While this issue is important in practice (e.g., when building 
real P2P systems), the cost of finding global information 
can be amortized over the cost of other operations, e.g., node searching.
For completeness, we discuss how to obtain this global information in
Section~\ref{sect:p2p}.

\section{Problem setup, cost model and the algorithm of Ganesan {\em
    et al.}}
\label{sect:model}

In this section, we describe the problem setup, discuss the cost of an
algorithm and review the algorithm of Ganesan {\em et al.}, the
{\adjload} algorithm.

\subsection{Problem setup}

We follow closely the basic setup of~\cite{GanesanBGM04-vldb}.  The
system consists of $n$ nodes and maintains a collection of keys. Let
$V$ be the set of all nodes. The key space is partitioned into $n$
ranges, with boundaries $R_0\leq R_1\leq\cdots\leq R_n$. Let $N_i$ be
the $i$-th node that manages a range $[R_{i-1},R_i)$. For any node $u\in
V$, let $L(u)$ be the number of keys stored in $u$. At any point of
time, there is an ordering of the nodes. This ordering defines {\em
  left} and {\em right} relations among nodes and this relation is
crucial to our analysis.

As in previous work, in some operation, a node requires non-local
information, namely the maximum and minimum loads and the
locations of the nodes with the maximum and minimum loads.

When a key is inserted or deleted, the node that manages the range
containing the key must update its data.  After that, the
load-balancing algorithm is invoked. Our goal is to maintain the ratio
of the maximum load to the minimum load, called the imbalance ratio.
We say that a load-balancing algorithm guarantees an \textit{imbalance
  ratio} $\sigma$ if after each insertion or deletion of a key and the
execution of the algorithm, $\max_u L(u)\leq \sigma\cdot\min_u
L(u)+c_0$ for some constant $c_0$.

We assume that, initially, each node has a small constant load, $c_0$.
As in~\cite{GanesanBGM04-vldb}, we ignore the concurrency issues and
consider only the serial schedule of operations.

\subsection{The cost}

To analyze the cost of a load-balancing algorithm, we follow the three
types of costs discussed in Ganesan {\em et al.}. 

\begin{enumerate}
\item {\bf Data Movement.} Each operation that moves a key from one
  node to another is counted as a unit cost.
\item {\bf Partition Change.}  When load-balancing steps are
  performed, the range of keys stored in each node may change.  This
  change has to be propagated through the system so that the next
  insertion or deletion goes to the right node.  
\item {\bf Load Information.} This work occurs when a node requests
  non-local information, e.g., requesting the node with the
  minimum or maximum load.
\end{enumerate}

In this paper, we analyze the amortized cost of the load-balancing
algorithm, i.e., we consider the worst-case cost of a sequence of $m$
load-balancing steps instead of a single one.

As in Ganesan {\em et al.}, we first analyze a simpler model that
accounts only for the data movement cost.  For partition change cost and load
information cost, we assume that there is a centralized server which
maintains the partition boundaries of each node and can answer the
request for global information. 

This assumption can be removed as shown in Ganesan {\em et al.}.  For
completeness, we discuss this in Section~\ref{sect:p2p}.

\subsection{The algorithm of Ganesan {\em et al.}}
\label{sect:GBGM}

We briefly review the algorithm introduced
in~\cite{GanesanBGM04-vldb}, the {\adjload} algorithm. The algorithm
uses two basic operations, {\sc NbrAdjust} and {\sc Reorder}
operations, defined as follows.

{\sc NbrAdjust}: \textit{Node $N_i$ transfers its load to its neighbor
  $N_{i+1}$. This may change the boundary $R_i$ of $N_i$ and
  $N_{i+1}$}.

{\sc Reorder}: \textit{ Consider a node $N_i$ with an empty range
  $\left[ R_i, R_i \right)$.  $N_i$ relocates its position and
  separates the range of $N_j$.  Then, the range $\left[ R_j, X
  \right)$ is managed by $N_j$, whereas the range $\left[ X, R_{j+1}
  \right)$ is managed by $N_i$ for some $X, R_j \leq X \leq
  R_{j+1}$. Finally, rename nodes appropriately.}

For some constant $c$ and $\delta$, they define a sequence of
thresholds $T_m=\left\lfloor c\delta^m \right\rfloor$, for all $m \geq
1$, used to trigger the {\adjload} procedure described later.  When
$\delta=2$, they call their algorithm the Doubling Algorithm.  
The {\adjload} procedure works as well when $\delta>\phi=\frac{(1+\sqrt{5})}{2}\approx 1.618$, the golden
ratio. They call their algorithm that operates at that ratio, the
Fibbing Algorithm. They prove that the {\adjload} procedure running on that ratio
would guarantee the imbalance ratio $\sigma$ of $\delta^3\approx 4.237$.

Given the threshold sequence, the {\adjload} procedure is as follows. When
a node $N_i$'s load crosses a threshold $T_m$, the load-balancing
algorithm is invoked on that node. Let $N_j$ be the lightly-loaded
neighbor of $N_i$. If the load of $N_j$ is not too high, the {\sc
  NbrAdjust} operation is applied, following by two recursive calls of 
{\adjload} on $N_i$ and $N_j$. Otherwise, it checks the load of the
minimum-loaded node $N_k$, and if the imbalance ratio is too high, tries to balance the
loads with the {\sc Reorder} operation as follows.
First, $N_k$ transfers its load to its lightly-loaded neighbor $N$.
Then, the {\sc Reorder} operation is invoked on $N_k$ and $N_i$, and
finally another call to the {\adjload} procedure is invoked on node $N$.

%
%

\section{The algorithm for the insert-only case}
\label{sect:algorithm-insertonly}

In this section, we present the algorithm for the insert-only case.
This case is simpler to analyze and provides general ideas on how to
deal with the general case. It is also of practical interest because
in many applications, as in a file sharing, deletions rarely occur.

We present the {\minbalance} procedure which uses the same primitive
operations, the {\sc NbrAdjust} and {\sc Reorder} operations. However,
it is simpler than the algorithm {\adjload} of Ganesan {\em et al.}
Notably, the {\minbalance} procedure is not recursive and performs
only a constant number of primitive operations. We prove the bound on
the imbalance ratio of 7.464 and show that the
amortized cost of an insertion is a constant.

\subsection{The MinBalance operation}

After an insertion occurs on any node $u$, the {\minbalance}
operation is invoked. If the load of $u$ is more than $\alpha$ times
 the load of the current minimum-loaded node $v$, node $v$ transfers
its entire load to one of its neighbor nodes, which has a lighter load. 
After that, $v$ halves the load of $u$. We call these steps,  {\minbalance} steps.
Note that this procedure requires  information about the minimum
load; the system must maintain this information.

\begin{algorithm} {\bf Procedure} {\minbalance} $(u)$
  \begin{algorithmic}[1]
    \label{alg:minbalance}
    \STATE Let $v$ be the minimum-loaded node in the system.
    \IF{$L(u) > \alpha\cdot L(v)$} 
    \STATE //{\minbalance} steps 
    \STATE Let $z$ be the lightly-loaded neighbor of $v$. 
    \STATE Transfer all keys of $v$ to $z$. 
    \STATE Transfer a half-load of $u$ to $v$,
    s.t., $L(u) = \left \lceil \frac{L(u)}{2}\right\rceil$ and $L(v) =
    \left\lfloor \frac{L(u)}{2} \right\rfloor$.  
    \STATE \{Rename nodes appropriately.\}
    \ENDIF
  \end{algorithmic}
\end{algorithm}

\subsection{Analysis of the imbalance ratio for the insert-only case}

We assume the notion of time of the system in a natural way.  For
simplicity, we assume that each operation completes instantly.

For any time $t$, let $L_t(u)$ and $L'_t(u)$ denote the load of node
$u$ after time $t$ and right before time $t$, i.e., $t-\epsilon$
for some $\epsilon>0$, respectively.  Let $Min_t$ and $Min'_t$ denote
the minimum load in the system after time $t$ and right before time
$t$ respectively, i.e., $Min_t =\min_{u\in V} L_t(u)$ and $Min'_t
=\min_{u\in V} L'_t(u)$.

We shall prove that the following invariant holds all the time.
\begin{quote} {\sc For any time $t$, the load of any node is not over
    $(\alpha+2)$ times the minimum load.}
\end{quote}
Note that the ratio of the maximum load to the minimum load is bounded by
$\alpha+2$, i.e., $\frac{Max_t}{Min_t}\leq (\alpha+2)$ where $Max_t$
is the maximum load at time $t$.  We  prove this invariant in
Theorem~\ref{thm:maxload}.

\subsubsection{Overview of the analysis}

While the analysis is rather involved, the idea is not very difficult to understand.  
This section gives an short overview to the analysis.

The main idea of the analysis is to prove that the imbalance ratio
remains under a constant after any operation.  We assume that the system starts with a uniform 
load distribution.   We first show the key property for the insert-only case, that is,
the minimum load never decreases.  Then, we analyze how the insertions change the loads of the nodes.

When an insertion occurs on any node $u$, there are two cases
depending on if $u$ calls   {\minbalance} steps or not.  
After insertion, the bound on the load of $u$ itself can be verified easily.

However, the insertion also affects two other nodes, i.e., 
the minimum loaded node $v$ and its lighter neighbor $z$.
When  {\minbalance} steps are invoked, the minimum loaded 
node $v$ transfers its entire load to $z$ and shares a half load of $u$.   
While showing the bound on the load of $v$ is pretty straight-forward,
the bound on $z$ requires more analysis because $z$
may receive loads many times through out the execution of the algorithm. 
The harder case for showing the bound on $z$ is when $z$ repeatedly takes
loads from its neighbors without any insertions to $z$.  However, in that case,
we know that $z$'s neighbors at some point become the minimum loaded node;
thus, we can establish the bound on the load of $z$ relative to the minimum load and
prove the required ratio.

\subsubsection{The analysis of the imbalance ratio}
First, we show an important property of the minimum load in
the system, i.e., after any operation, the minimum load never
decreases.
\begin{lemma}
\label{lem:minload_not_decrease} 
Suppose $\alpha > 2$. Then the minimum load of the system never decreases.
\end{lemma}

\begin{proof}
  Inserting keys into the system only increases load.  The only steps
  that decrease loads are the  {\minbalance} steps; therefore, we
  consider only these steps.  For each {\minbalance} steps, there are
  three nodes involved, i.e., node $u$ which initiates the 
  {\minbalance} steps at some time $t$, the current minimum-loaded
  node $v$ and node $z$ which is the lightly-loaded neighbor of $v$ at
  time $t$.  First, $v$ transfers its entire keys to $z$; hence $z$'s
  load increases.  The new load of each of $u$ and $v$ is a half of
  $u$'s current load.  Since $u$ invokes the  {\minbalance} steps, its
  load must more than $\alpha\cdot Min_t > 2\cdot Min_t$.  Thus, we
  have that after the steps, $u$'s load and $v$'s load are at least $Min_t$
  as required. \qed
\end{proof}

After an insertion occurs on any node $u$, the load of $u$  changes.
Next lemma guarantees a good ratio on node $u$.

\begin{lemma}
\label{lem:insert-property} 
Consider an insertion occurring on node $u$ at some time $t$.  Suppose
$\alpha\geq 3$. After an insertion and its corresponding
load-balancing steps, $L_t(u)\leq \alpha \cdot Min_t$.
\end{lemma}

\begin{proof}
  After an insertion, we consider two cases.  The first case is when the
  {\minbalance} steps are invoked. After that the load of $u$ decreases by half. 
  Note that the load of $u$ after insertion cannot be over
  $(\alpha+2)\cdot Min'_t+1$ which comes from the invariant and a key
  insertion. Then, we have that
  \[L_t(u)\leq\left\lceil\frac{(\alpha+2)\cdot
      Min'_t+1}{2}\right\rceil\leq\frac{(\alpha+2)\cdot
    Min'_t+Min'_t}{2}.\]  
  For any $\alpha\geq 3$, it follows that $\frac{(\alpha+3)}{2}\leq\alpha$. 
  Hence, $L_t(u)\leq \alpha\cdot Min'_t$.
  From the non-decreasing property, we have $L_t(u)\leq \alpha\cdot Min_t$.
  
  We are left to consider the second case when the {\minbalance} steps
  are not invoked.  From the algorithm, the load of $u$ in this case
  is not over $\alpha\cdot Min'_t$.  Because the minimum load never
  decreases, we have $L_t(u)\leq \alpha\cdot Min_t$.
\end{proof}

Note that the {\minbalance} procedure ``sees'' the load of the newly inserted node
and the minimum load, while it ignores the load of node $z$, the
previous neighbor of the minimum-loaded node. We define the
{\em min-transfer} event, i.e., we say that a min-transfer event occurs
on node $z$, when the minimum-loaded node transfers its load to $z$,
its lightly-loaded neighbor. Most of our analysis deals with the load
of nodes suffered from this kind of transfer.

Consider a sequence of min-transfer events occurring on node $z$.
Let $t_i$ represent the time after the $i$-th min-transfer event occurs
on $z$. Note that before a min-transfer event occurs on $z$, an insertion may
occur on it.  Next lemma shows that the load of $z$ has a good ratio
in this case.

\begin{lemma}
  \label{lem:insert-between-min-transfer}
Suppose that an insertion occurs on node  $z$ at some time $t$
right before the $i$-th min-transfer event occurs on $z$. Then
$L_{t_{i}}(z)\leq(\alpha + 1)\cdot Min_{t_{i}}$.
\end{lemma}

\begin{proof}
  From Lemma~\ref{lem:insert-property}, after an insertion occurs
  on $z$ at some time $t$, $L_{t}(z)\leq\alpha\cdot Min_{t}$. The load of $z$
  increases again when the $i$-th min-transfer event occurs  at
  time $t_i$.  Thus,
  \[
  L_{t_{i}}(z)=L_t(z)+ Min'_{t_{i}} \leq \alpha\cdot Min_{t} +
  Min'_{t_{i}}.
  \]
  Therefore, we have that $L_{t_{i}}(z) \leq (\alpha+1)\cdot
  Min_{t_i}$ because of the non-decreasing property of the minimum load.
\end{proof}

As the result of Lemma~\ref{lem:insert-between-min-transfer}, later
when dealing with the load of node $z$, we only have to consider the case
that no insertion occurs on $z$ before a min-transfer event occurs
on it.

In our analysis, we categorize a min-transfer event into two types:
the {\em left-transfer} event and the {\em right-transfer} event. The
left-transfer event on $z$ is a min-transfer event that $z$ gets the
keys from the node, which is to the left of $z$ (see
Figure~\ref{fig:1} (a)). The right-transfer event on $z$ is
a min-transfer event that $z$ gets the keys from the node, which is
to the right of $z$ (see Figure~\ref{fig:1} (b)).

\begin{figure}
  \centering \includegraphics[width=0.35\textwidth]{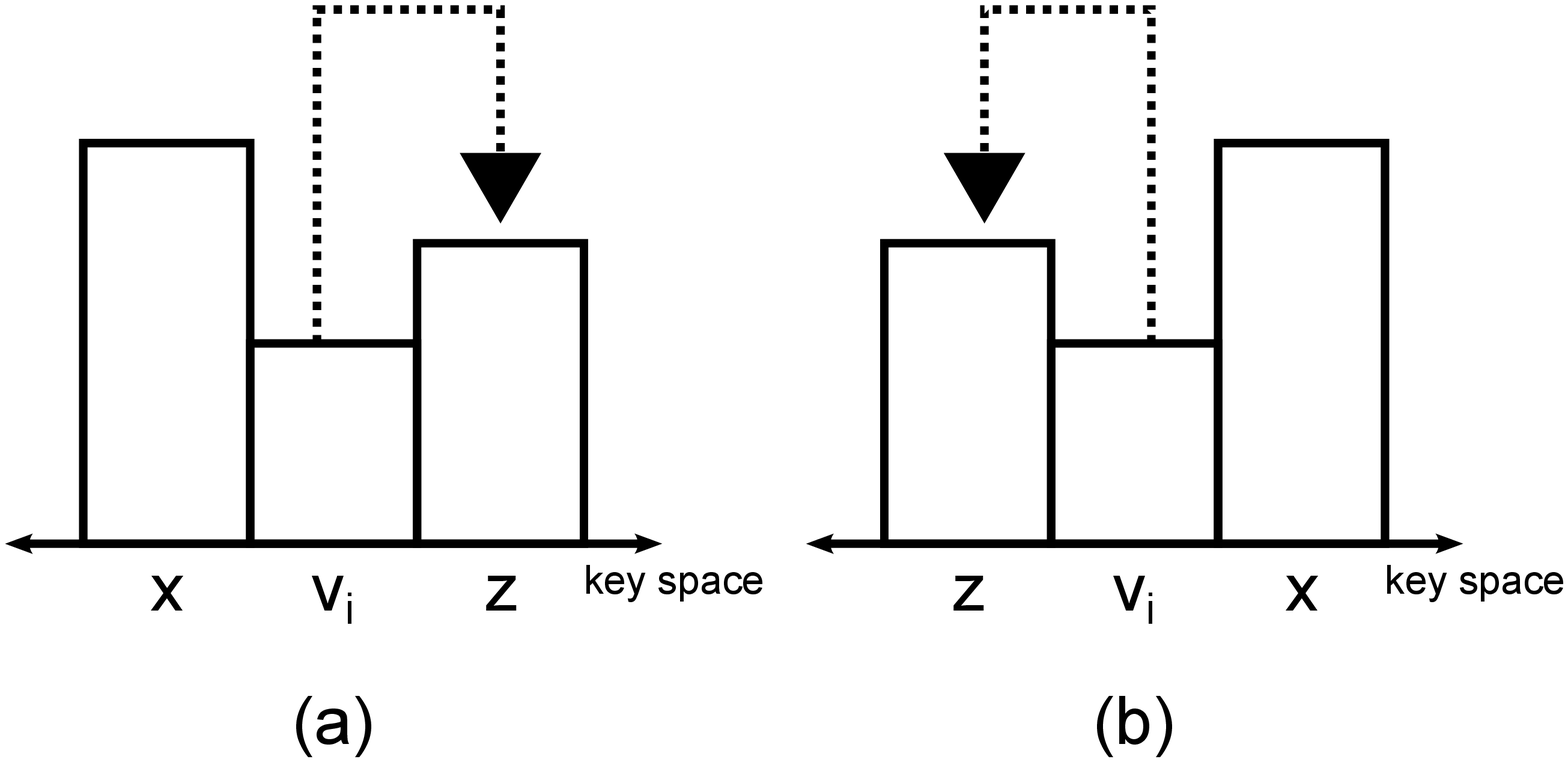}
  \caption{(a) The left-transfer event on z. (b) The right-transfer
    event on z.}
  \label{fig:1}
\end{figure}

When a new min-transfer event occurs on $z$, there are two
situations, i.e., the new min-transfer event is the same type as the
previous min-transfer event and the new min-transfer event is the
other type.

The next lemma bounds the load of $z$ when $l$ min-transfer events of
the same type occur on $z$ (probably not one after another).  Note
that, it is straightforward to show that the bound of $(\alpha+l+2)$
but we need a better bound.

\begin{lemma}
  \label{lem:genneralizemintransfer}
  Suppose $\alpha \geq 1+\sqrt{5}\approx 3.237$ and $l\geq
  0$. Consider the $i$-th, $(i+1)$-th, ..., $(i+l+1)$-th min-transfer
  events on $z$. If the $i$-th and  $(i+l+1)$-th
  min-transfer events are of the same type, while the $(i+1)$-th,
  $(i+2)$-th, ..., $(i+l)$-th min-transfer events are of type
  different from that of the $i$-th and  $(i+l+1)$-th events, then
  after the $(i+l+1)$-th min-transfer event, $L_{t_{i+l+1}}(z)\leq
  (\alpha+l+1)\cdot Min_{t_{i+l+1}}$.
\end{lemma}

\begin{proof}
  Without loss of generality, we assume that the $i$-th and 
  $(i+l+1)$-th min-transfer events occurring on $z$ are the
  right-transfer event; and the $(i+1)$-th, $(i+2)$-th, ...,
  $(i+l)$-th are the left-transfer event.  Let $v_{i}$ be the
  minimum-loaded node at time $t_i$.

  We first deal with the case that there exists an insertion occurring
  on $z$ between the $i$-th and  $(i+l+1)$-th min-transfer events.
  Assume that the latest insertion occurs on $z$ right before the
  $k$-th min-transfer event where $i< k\leq i+l+1$. From
  Lemma~\ref{lem:insert-between-min-transfer}, after the $k$-th
  min-transfer event, the load of $z$ is not over $\alpha +1$ times 
  the minimum load at time $t_k$, i.e.,
	
  \[
  L_{t_k}(z)\leq (\alpha+1)\cdot Min_{t_k}.
  \]

  Since no insertion occurs on $z$ after the $k$-th min-transfer
  event, we have 
	
  \[
  L_{t_{i+l+1}}(z)\leq (\alpha+1)\cdot Min_{t_k} + Min_{t_{k+1}}
  +\cdots + Min_{t_{i+l+1}}.
  \]
	
  Because the minimum load never decreases, it follows that
	
  \[
  L_{t_{i+l+1}}(z) \leq (\alpha+1)\cdot Min_{t_{i+l+1}} + (i+l-k)\cdot
  Min_{t_{i+l+1}},
  \]

  and finally we have $L_{t_{i+l+1}}(z)\leq (\alpha+l+1)\cdot
  Min_{i+l+1}$, because $k> i$.
		
  We are left to consider the case that there is no insertion into $z$
  between the $i$-th and  $(i+l+1)$-th min-transfer events.  We
  make the first claim.
    
  \begin{claim}\label{clamin:1}
  \[
  L_{t_{i+l+1}}(z) \leq L_{t_{i}}(z) + (l+1)\cdot Min'_{t_{i+l+1}}.
  \]
\end{claim}
\begin{proof}
  To prove the claim, we note that for any $i< j\leq i+l+1$, the load
  of $z$ right before the $j$-th min-transfer event is equal to its
  load after the $(j-1)$-th min-transfer event, i.e.,

  \begin{equation}\label{eqn:l'_l}
    L'_{t_j}(z)=L_{t_{j-1}}(z).
  \end{equation}  

  Also, for any $i< j\leq i+l+1$, the load of $z$ after the
  $j$-th min-transfer event increases by the minimum load before time
  $t_j$, $Min'_{t_j}$.  Thus,

  \begin{equation} \label{eqn:def_load_after_mintransfer}
    L_{t_j}(z)=L'_{t_j}(z)+Min'_{t_j},
  \end{equation}  

  and from Eq.~(\ref{eqn:l'_l}),

  \begin{equation} \label{eqn:tele}
    L_{t_j}(z)=L_{t_{j-1}}(z)+Min'_{t_j}.
  \end{equation}

  Telescoping, we have
  \[
  L_{t_{i+l+1}}(z) \leq L_{t_{i}}(z) +Min'_{t_{i+1}}+Min'_{t_{i+2}}+
  \cdots+ Min'_{t_{i+l+1}} \leq L_{t_{i}}(z) + (l+1)\cdot
  Min'_{t_{i+l+1}}.
  \]
  The last step is from the non-decreasing property of the minimum
  load; and the claim follows.
\end{proof}

\begin{figure}
    \centering \includegraphics[width=0.50\textwidth]{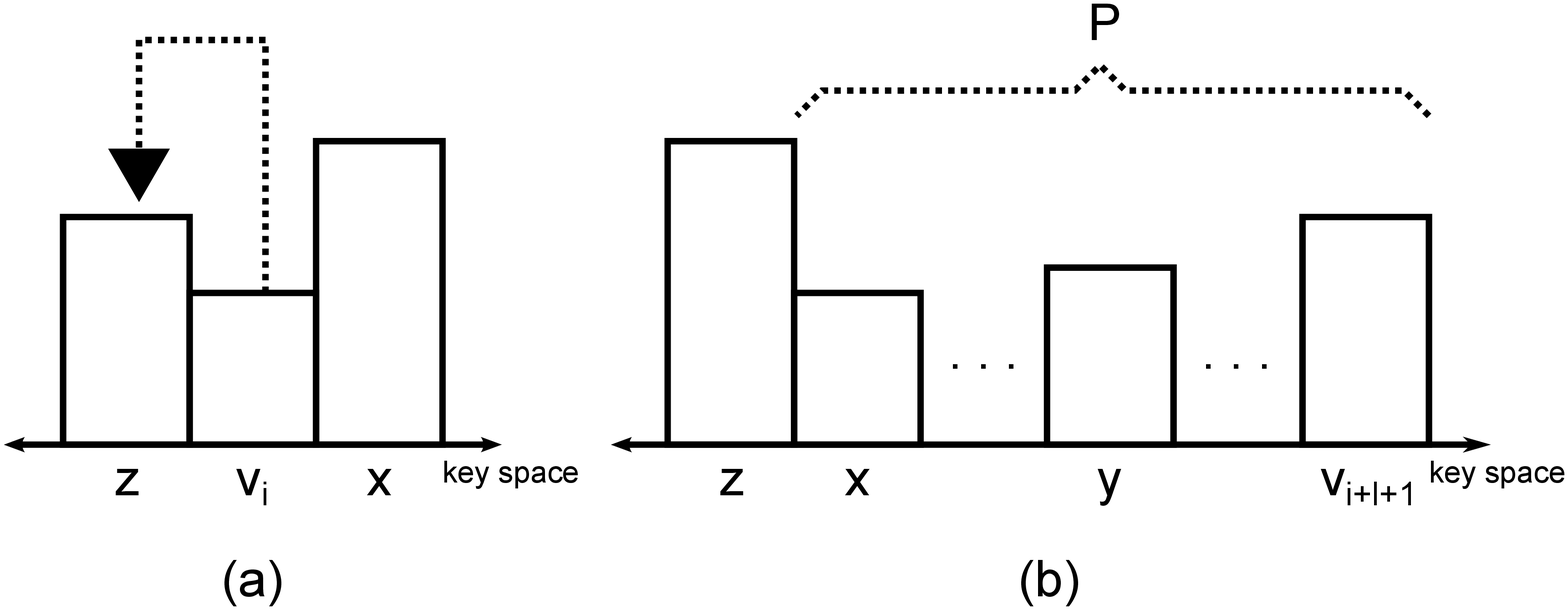}
    \caption{(a) The $i$-th min-transfer event on z. (b) The node set
      $P$ between $x$ and $v_{i+l+1}$.}
    \label{fig:2}
  \end{figure}

  Now, consider the $i$-th min-transfer event on $z$. Recall that the 
  $i$-th min-transfer event is a right-transfer event.  Let $x$ be the node on
  the right of $v_i$ right before the $i$-th min-transfer event occurs
  on $z$ (see Figure~\ref{fig:2} (a)).  Note that because the
  $(i+l+1)$-th min-transfer event is a right-transfer event, $x$
  must become the minimum-loaded node at some point after $t_i$. Let $t^*$ be the
  time that $x$ becomes the minimum-loaded node  after $t_i$.
  Node $x$ plays a crucial role in our analysis.
  
  There are two cases depending on $x$ calls {\minbalance} steps or not.
  
  \textbf{Case 1:} $x$ does not call {\minbalance} steps between time
  $t_i$ and $t_{i+l+1}$.
  
  In the $i$-th min-transfer event occurring on $z$, $v_i$ transfers
  its entire load to $z$ at time $t_i$, but not $x$. That means the
  load of $z$ right before the $i$-th min-transfer event is not over
  the load of $x$, i.e., $L'_{t_i}(z)\leq L'_{t_i}(x)$. Then, after
  time $t_i$,

  \begin{equation}\label{eqn:1.5}
    L_{t_i}(z) = L'_{t_i}(z)+Min'_{t_i}\leq L'_{t_i}(x)+Min'_{t_i},
  \end{equation}

  and from the Claim~\ref{clamin:1},
  \begin{equation}\label{eqn:1.6}
    L_{t_{i+l+1}}(z) \leq L'_{t_i}(x)+Min'_{t_i} + (l+1) \cdot
    Min'_{t_{i+l+1}}.
  \end{equation}
  
  After time $t_i$, $x$ becomes the minimum-loaded node at  time $t^*$.
  Note that $x$ does not call the  {\minbalance} steps. Then, the
  event that may occur on $x$ after time $t_i$ is an insertion or a
  min-transfer event. Thus, $x$'s load does not decrease. We have
  $L'_{t_i}(x)\leq Min_{t^*}$. Then,
  
  \[L_{t_{i+l+1}}(z) \leq Min_{t^*}+Min'_{t_i} + (l+1)\cdot
  Min'_{t_{i+l+1}}. \]
  
  From the non-decreasing property of the minimum load, we have
  $L_{t_{i+l+1}}(z) \leq (l+3)\cdot Min_{t_{i+l+1}}$ and when $\alpha
  \geq 2$, it follows that $L_{t_{i+l+1}}(z)\leq (\alpha+l+1)\cdot Min_{t_{i+l+1}}$.

  \textbf{Case 2:} $x$ calls {\minbalance} steps between time $t_i$
  and $t_{i+l+1}$.
  
  Let $t'$ be the latest time that $x$ calls  {\minbalance}
  steps. Since {\minbalance} steps are invoked, some insertion must occur
  on $x$. After that, $x$'s load is over $\alpha$ times  the minimum
  load at time $t'$. Then, the load of $x$ is divided into halves,
  i.e.,
  \[ L_{t'}(x)\geq \left\lceil \frac{\alpha
      Min_{t'}}{2}\right\rceil.\]
  
  Note that after time $t'$, $x$ does not call {\minbalance} steps
  and it becomes the minimum-loaded node at  time $t^*$. The event that
  may occur on $x$ after time $t'$ is an insertion or a min-transfer
  event. Thus, $x$'s load after time $t'$ does not decrease. Therefore,
  $L_{t'}(x)\leq L_{t^*}(x)$.  Moreover, we know that $L_{t^*}(x)\leq
  Min'_{t_{i+l+1}}$ by the non-decreasing property of the minimum
  load. Then,
  
  \begin{equation}\label{eqn:1.8}
    Min'_{t_{i+l+1}}\geq \left\lceil \frac{\alpha
        Min_{t'}}{2}\right\rceil.
  \end{equation}  
 
  From the invariant, the load of $z$ at time $t_i$ is not
  over $(\alpha+2)$ times  the minimum load at time $t_i$. We know
  that the minimum load cannot decrease; we have

  \begin{equation}\label{eqn:1.10}
    L_{t_i}(z)\leq (\alpha+2)\cdot Min_{t_i}\leq (\alpha+2)\cdot
    Min_{t'}.
  \end{equation}
  
    Consider the Claim~\ref{clamin:1}. We divide it by $Min_{i+l+1}$,
  i.e.,

  \begin{eqnarray*}
    \frac{L_{t_{i+l+1}}(z)}{Min_{i+l+1}} &\leq& \frac{L_{t_i}(z)+
      (l+1)\cdot Min'_{t_{i+l+1}}}{Min_{t_{i+l+1}}}.\\
  \end{eqnarray*}
  
  By the non-decreasing property of the minimum load,
  $Min'_{t_{i+l+1}}\leq Min_{t_{i+l+1}}$. We have
  \begin{eqnarray*}\label{eqn:L/M}
    \frac{L_{t_{i+l+1}}(z)}{Min_{i+l+1}} &\leq& \frac{L_{t_i}(z)+
      (l+1)\cdot Min'_{t_{i+l+1}}}{Min'_{t_{i+l+1}}}\\
    & = & (l+1)+\frac{L_{t_i}(z)}{Min'_{t_{i+l+1}}},
  \end{eqnarray*}
 	
  and from Eq.~(\ref{eqn:1.8}) and~(\ref{eqn:1.10}),

  \begin{eqnarray*}
    \frac{L_{t_{i+l+1}}(z)}{Min_{i+l+1}} & \leq &
    (l+1)+\frac{(\alpha+2)\cdot Min_{t'}}{\frac{\alpha
        Min_{t'}}{2}}\\
    &=& (l+1)+\frac{2(\alpha+2)}{\alpha}.              
  \end{eqnarray*}
  
  From the assumption that $\alpha\geq 1+\sqrt{5}$, it follows that
  $\frac{2(\alpha+2)}{\alpha}\leq\alpha$. Then, we have
  \[
  \frac{L_{t_{i+l+1}}(z)}{Min_{i+l+1}}\leq (\alpha+l+1).
  \]
  \end{proof}

Using previous lemmas, we can conclude the bound on the load of $z$
after a min-transfer event occurring on it.

\begin{lemma}
  \label{lem:min-transfer}
  Suppose $\alpha \geq 1+\sqrt{5} \approx 3.237$. After the $i$-th
  min-transfer event occurs on $z$, $ L_{t_i}(z)\leq(\alpha + 2)\cdot
  Min_{t_i}.  $
\end{lemma}

\begin{proof}
  We prove by induction on the number of min-transfer events occurring
  on $z$.
  
  For the base case, the first min-transfer event occurs on $z$.  If
  some insertion occurs on it before the min-transfer event, t
  from Lemma~\ref{lem:insert-between-min-transfer}, 
  $L_{t_1}(z)\leq(\alpha+1)\cdot Min_{t_1}$. We are left to consider the
  case that no insertion occurs on $z$ before the first min-transfer
  event.  The load of $z$ before the min-transfer event is equal to its
  load at the beginning of the system, $c_0$. Note that the minimum
  load at time $t'_1$ is also $c_0$. After the first
  min-transfer event, $L_{t_1}(z) = 2Min_{t_1}$. Thus, the load of $z$ 
   in both cases is not over $(\alpha+2)\cdot Min_{t_1}$.
    
  Assume that the load of any node at time $t_k$ is not over $(\alpha
  +2)\cdot Min_{t_k}$. Consider the $(k+1)$-th min-transfer event on
  $z$. There are two cases.

  \begin{figure}
    \centering \includegraphics[width=0.75\textwidth]{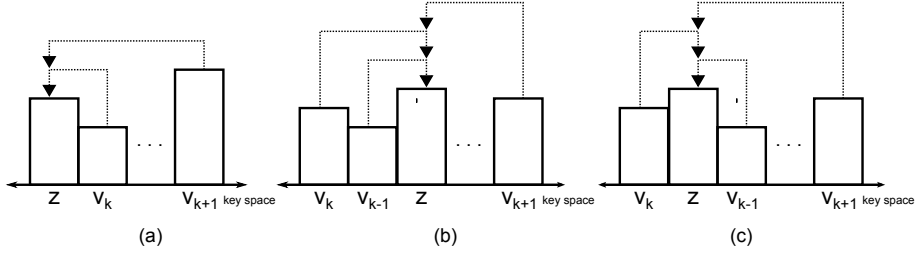}
    \caption{(a) the $k$-th and $(k+1)$-th min-transfer events are the same
      type. (b) the $(k-1)$-th and $k$-th min-transfer events are the same
      type but the $(k+1)$-th min-transfer event is a different type.
      (c) the $(k-1)$-th and $(k+1)$-th min-transfer events are the same
      type but the $k$-th min-transfer event is a different type.}
    \label{fig:3}
  \end{figure}

  \textbf{Case 1:} The $(k+1)$-th and $k$-th min-transfer events
  are the same type (see Figure~\ref{fig:3} (a)).  From
  Lemma~\ref{lem:genneralizemintransfer}, when set $l=0$, after the
  $(k+1)$-th min-transfer event, $L_{t_{k+1}}(z)\leq(\alpha+1)\cdot Min_{t_{k+1}}$. 
  Thus, the lemma holds in this case.

  \textbf{Case 2:} The $(k+1)$-th min-transfer event is a different type
  from the $k$-th min-transfer event. If there is any insertion into
  $z$ between time $t_{k}$ and $t_{k+1}$, from
  Lemma~\ref{lem:insert-between-min-transfer}, 
  $L_{t_{k+1}}(z)\leq(\alpha+1)\cdot Min_{t_{k+1}}$.  We consider the $(k-1)$-th
  min-transfer event.  There are three sub cases.

  \textbf{Sub case 2.1:} There is no the $(k-1)$-th min-transfer event;
  hence, the $k$-th min-transfer event is the first min-transfer event. 
  We can bound the load of $z$ after $t_k$ like the base case,
  i.e., $L_{t_{k}}(z) \leq (\alpha+1)\cdot Min_{t_{k}}$. The load of $z$
  increases again when the $(k+1)$-th min-transfer event occurs on
  it at time $t_{k+1}$, i.e.,
  \[ L_{t_{k+1}}(z)= L_{t_{k}}(z) +  Min'_{t_{k+1}} \leq (\alpha+1)\cdot Min_{t_{k}}+ Min'_{t_{k+1}}.\]
  From the non-decreasing property of the minimum load, we have that $
  L_{t_{k+1}}(z) \leq (\alpha+2)\cdot Min_{t_{k+1}}$. Thus, the lemma
  holds in this sub case.

  \textbf{Sub case 2.2:} The $(k-1)$-th min-transfer event is the same
  type as the $k$-th min-transfer event (see Figure~\ref{fig:3} (b)).
  From Lemma~\ref{lem:genneralizemintransfer}, when set $l=0$, after the
  $k$-th min-transfer event, $L_{t_{k}}(z)\leq(\alpha+1)\cdot Min_{t_{k}}$. The load of  $z$ increases again by the
  $(k+1)$-th min-transfer event, i.e.,
  \[ L_{t_{k+1}}(z)=  L_{t_{k}}(z) + Min'_{t_{k+1}}\leq (\alpha+1)\cdot Min_{t_{k}}+ Min'_{t_{k+1}}.\]
   From the non-decreasing property of the minimum load,
  we have $L_{t_{k+1}}(z)\leq (\alpha+2)\cdot Min_{t_{k+1}}$.  Thus,
  the lemma holds in this sub case.

  \textbf{Sub case 2.3:} The $(k-1)$-th min-transfer event is a
  different type from the $k$-th min-transfer event, i.e., the
  $(k-1)$-th min-transfer event is the same type as the $(k+1)$-th
  min-transfer event (see Figure~\ref{fig:3} (c)). From
  Lemma~\ref{lem:genneralizemintransfer}, when set $l=1$,  
  after the $(k+1)$-th min-transfer event, $L_{t_{k+1}}(z)\leq(\alpha +
  2)\cdot Min_{t_{k+1}}$ and thus, the lemma holds in this sub case.
\end{proof}

We are ready to prove the invariant. Note that a load of any node
may change from insertion or min-transfer event.  From
Lemma~\ref{lem:insert-property} and~\ref{lem:min-transfer}, we can
conclude the invariant, which guarantees the bound of any load in the
system.

\begin{theorem}\label{thm:maxload}
  Consider the insert-only case. Suppose $\alpha \geq 1+\sqrt{5}
  \approx 3.237$. For any node $u \in V$, after any event at time
  $t$, $ L_t(u)\leq(\alpha+2)\cdot Min_t.  $
\end{theorem}
	
In our load-balancing algorithm, we want to minimize the number of
moving keys and to guarantee a constant imbalance ratio.  Imbalance
ratio, $\sigma$, is defined as the ratio of the maximum to minimum
load in the system. We show that the imbalance ratio from our
algorithm in the insert-only case is a constant.

\begin{corollary}\label{thm:Imbalance_ratio}
  Consider the insert-only case. Suppose $\alpha \geq 1+\sqrt{5}
  \approx 3.237$. The imbalance ratio of the algorithm is a constant.
\end{corollary}

\begin{proof}
  Prove directly from Theorem~\ref{thm:maxload}.
\end{proof}

\subsection{Cost of the algorithm in insert-only case}
\label{sect:imbalanceratioproof}

Our algorithm uses two operations, i.e., insert and {\minbalance} operations.
 We consider the cost of each operation. Moving a single
key from one node to another is counted as a unit cost.  We follow the
analysis in~\cite{GanesanBGM04-vldb} based on the potential function
method, and use the same potential function.

\begin{theorem}\label{thm:amortized_cost_insertonly}
   Suppose that $\alpha>2(1+\sqrt{3})\approx 5.464$. The amortized costs of our algorithm in the insert-only case are
  constant.
\end{theorem}

\begin{proof}
  Let $\bar{L}_t$ denote the average load at time $t$ and let $N_i$ be
  the $i$-th node that manages a range $[R_{i-1}, R_i)$. We consider the
  same potential function as~\cite{GanesanBGM04-vldb}, i.e., $\Phi =
  \frac{c\left(\sum^{n}_{i=1}(L_t(N_i))^2\right)}{\bar{L}_t}$, where $c$ is a
  constant to be specific later. We show that the gain in potential
  when an insertion occurs is at most a constant and the drop in
  potential when a {\minbalance} operation occurs pays for the cost
  of the operation.

  \textbf{Insertion:} Consider an insertion of a key occurring on node
  $N_j$ at time $t$ before any load-balancing steps are invoked. Note
  that, the load of all nodes except $N_j$ does not change during the
  insertion. Thus, the gain in potential, $\Delta\Phi$, is

  \begin{eqnarray*}
    \Delta\Phi
    &=&\frac{c\left(\sum^{n}_{i=1}(L_t(N_i))^2\right)-c\left(L_t(N_j)\right)^2+c\left(L_t(N_j)+1\right)^2}{\bar{L}_{t}+\frac{1}{n}}-\frac{c\left(\sum^{n}_{i=1}(L_t(N_i))^2\right)}{\bar{L}_{t}}\\
    &\leq&\frac{c\left((L_{t}(N_j)+1)^2\right)-c\left(L_t(N_j)^2\right)}{\bar{L}_{t}}\\
    &=& \frac{c\left(2L_{t}(N_j)+1\right)}{\bar{L}_t}.
  \end{eqnarray*}	

  From the invariant,  $L_{t}(N_j)\leq (\alpha+2)\cdot Min_t$. Since 
  $\bar{L}_t\geq Min_t\geq 1$, then
  $L_{t}(N_j) \leq (\alpha+2)\cdot \bar{L}_t$. Hence, 
  
  \begin{eqnarray*}
    \Delta\Phi &\leq&
    \frac{c\left(2(\alpha+2)\cdot \bar{L}_t+1\right)}{\bar{L}_t}
    \leq c\left(2\cdot \alpha+5\right).\\
  \end{eqnarray*}	

  Since an insertion moves a new key to some node, the actual cost of an
  insertion is a unit cost. Thus, the amortized cost of an insertion is
  bounded by a constant.

  \textbf{MinBalance:} There are three nodes involved, i.e.,
  $N_j$ which calls {\minbalance} steps, the minimum-loaded node $N_k$
  and the $N_k$'s lightly-loaded neighbor $N_l$. When node $N_j$ calls
  {\minbalance} steps, $N_k$ transfers its entire load to $N_l$. After
  that, $N_k$ shares a half the load of $N_j$.  The drop potential is

  \begin{eqnarray*}
    \Delta\Phi&=& \frac{c\left(L_t(N_j)^2+L_t(N_k)^2+L_t(N_l)^2
        -2(\frac{L_t(N_j)}{2})^2-(L_t(N_k)+L_t(N_l))^2\right)}{\bar{L}_t}\\
    &=&\frac{c\left(\frac{L_t(N_j)^2}{2}-2L_t(N_k)L_t(N_l)\right)}{\bar{L}_t}.
  \end{eqnarray*}	

  From the invariant, we have $L_t(N_l)\leq(\alpha+2)\cdot L_t(N_k)$.
  Since $N_j$ calls the  {\minbalance} steps, we know that
  $L_t(N_k)\leq \frac{L_t(N_j)}{\alpha}$. We have  $L_t(N_l)\leq (\alpha+2)\cdot
  \frac{L_t(N_j)}{\alpha}$.  Then, 

  \begin{eqnarray*}
    \Delta\Phi & \geq &
    \frac{c\left(\frac{L_t(N_j)^2}{2}-\frac{2(\alpha+2)\cdot
          L_t(N_j)^2}{\alpha^2}\right)}{\bar{L_t}}\\ & = &
    \frac{c\left(L_t(N_j)^2(\frac{1}{2} - \frac{2(\alpha+2)
        }{\alpha^2})\right)}{\bar{L_t}}.
  \end{eqnarray*}

  Again, from the invariant, we have $ \bar{L_t}\leq(\alpha+2)\cdot
  L_t(N_k)\leq(\alpha+2)\cdot L_t(N_j)$.  Then,

  \begin{eqnarray*}
    \Delta\Phi&\geq& c L_t(N_j)\left(\frac{1}{2(\alpha+2)} -
      \frac{2}{\alpha^2}\right).
  \end{eqnarray*}

  The data movement of {\minbalance} steps is $\left\lfloor
    \frac{L_t(N_j)}{2} \right\rfloor$+ $L_t(N_k)<L_t(N_j)$.  For any $c
  >\left(\frac{2\alpha^2(\alpha+2)}{\alpha^2-4(\alpha+2)}\right)$ and $\alpha>2(1+\sqrt{3})\approx 5.464$, we
  have $\Delta\Phi\geq L_t(N_j)$.  Thus, the data movement cost can be
  paid by this drop in potential.
\end{proof}

%
%

\section{The algorithm for the general case}
\label{sect:algorithm-delete}

In this section, we consider the general case that supports both insert and delete
operations. In order to deal with the imbalance ratio, after these operations,
load balancing steps are invoked. For insertion, we perform the
{\minbalance} operation presented in the previous section. On the other
hand, for deletion, we present the {\split} operation.

\subsection{The Split operation}

The {\split} operation is invoked after a deletion occurs on any node
$u$. Let $z$ be the lightly-loaded neighbor of $u$. 
The load-balancing steps are called when the load of $u$ is less
than $\frac{1}{\beta}$ fractions of the maximum load at that time. There are
two types of load-balancing steps depending on the load of $z$.  If
the load of $z$ is more than $\frac{2}{\beta}$ fractions of the maximum load,
$u$ averages out its load with $z$. We call these steps, the
{\splitnbr}. In other case, $u$ transfers its entire load to $z$.
After that, $u$ shares a half load of the maximum-loaded node.  These
steps are called the {\splitmax}. Note that the {\splitnbr} calls
only the {\sc NbrAdjust} operation but the {\splitmax} calls both the
{\sc NbrAdjust} and the {\sc Reorder} operations.

We again note that to be able to perform these operations, the system
must maintain non-local information, i.e., the maximum load.

\begin{algorithm} {\bf Procedure} {\split} $(u)$
  \label{alg:split}
  \begin{algorithmic}[1]
    \STATE Let $w$ be the maximum-loaded node in the system. 
    \IF {$L(u) \leq \frac{L(w)}{\beta}$} 
    \STATE Let $z$ be the lightly-loaded neighbor of $u$. 
    \IF {$L(z)\leq 2\cdot \frac{L(w)}{\beta}$} 
    \STATE //{\splitmax} steps 
    \STATE Transfer all keys of $u$ to $z$.
    \STATE Transfer a half-load of $w$ to $u$, s.t., $L(w) =
    \left\lceil \frac{L(w)}{2}\right\rceil$ and $L(u) = \left\lfloor
      \frac{L(w)}{2}\right\rfloor$. 
    \ELSE 
    \STATE //{\splitnbr} steps 
    \STATE
    Move keys from $z$ to $u$ to equalize load, s.t., $L(u) = \left
      \lceil \frac{(L(u)+L(z))}{2}\right\rceil$ and $L(z) = \left\lfloor
      \frac{(L(u)+L(z))}{2} \right\rfloor$.
    \ENDIF
    \STATE \{Rename nodes appropriately.\}
    \ENDIF
  \end{algorithmic}
\end{algorithm}

\subsection{Analysis of imbalance ratio for the general case}

We assume the notion of time and the invariant in the same way as the
previous section. We analyze the imbalance ratio after any event.
In previous section, we analyze the ratio after insertion and
min-transfer event. In this section, we have to deal with deletion
and two more events:
\begin{itemize}
\item The \textit{nbr-transfer} event on node $z$ is the event that
  occurs when $z$ receives load from its neighbor, which invokes the
  {\splitmax} steps, and
\item The \textit{nbr-share} event on node $z$ is the event that occurs
  when $z$ shares load with its neighbor, which invokes the
  {\splitnbr} steps.
\end{itemize}

Let $Max_t$ and $Max'_t$ denote the maximum load in the system after
time $t$ and right before time $t$ respectively, i.e., $Max_t
=\max_{u\in V} L_t(u)$ and $Max'_t =\max_{u\in V} L'_t(u)$.

\subsubsection{Overview of the analysis}
The major problem for applying the proofs in the insert-only case to
the general case is the assumption that the minimum load cannot
decrease over time.  To handle this, we shall analyze the system in
phases.  Each phase spans the period where the minimum load is non-decreasing;
this allows us to apply mostly the same techniques to the analyze the situation 
when the update does not change the analysis phase.

The only way an update could cause phase change is when there is a deletion
in the minimum loaded node.  In that case, we show that if there is a deletion 
on the minimum loaded node that starts the phase change, the ratio between 
the minimum load and the maximum load is bounded by a constant.  This provides
the sufficient initial condition for the analysis of the next phase.

\subsubsection{Transition between phases}

We shall analyze the transition between two consecutive phases
occurring when the minimum load decreases. 

Consider a deletion occurring on
any nodes. Next lemma shows the load property of any node after
deletion occurring on it.

\begin{lemma}\label{lem:delete-property}
  Suppose $\beta>3$. Consider the case that a deletion occurs on
  some node $u$ at some time $t$.  After deletion and its corresponding
  load-balancing steps, $L_t(u)>\frac{Max_t}{\beta}$. Moreover, the minimum
  load can decrease in the case that  load-balancing steps are not
  invoked.
\end{lemma}

\begin{proof}
  After a deletion occurs on $u$, there are two cases.  The first case is when 
  load-balancing steps are not invoked. In this case, the load of
  $u$ after deletion is more than $\frac{Max'_t}{\beta}$.  Note that at time
  $t$, the only event that occurs in the system is a deletion on $u$.
  Then, the maximum load does not increase, i.e., $Max'_t\geq
  Max_t$. We have that $L_t(u)>\frac{Max'_t}{\beta}\geq \frac{Max_t}{\beta}$.
  
  We are left to consider the second case when  load-balancing
  steps are invoked.  In this case, the load of $u$ is not over
  $\frac{Max'_t}{\beta}$. Let $z$ be the lightly-loaded neighbor of
  $u$ at time $t$. There are two sub cases.

  \textbf{Sub case 1:} Node $u$ calls {\splitmax} steps.  First, 
  $u$ transfers its entire load to $z$.  After that, $u$ shares a half 
  load of the maximum-loaded node, i.e., 
  \[ L_t(u)= \left\lfloor \frac{Max'_t}{2}\right\rfloor > \frac{Max'_t}{\beta}.\]
  In this sub case, the node except $u$
     that its load increases at time $t$ is $z$. Note
  that, $z$ receives load from $u$, i.e., 
  \[ L_t(z)= L'_t(z)+L'_t(u)\leq \frac{2Max'_t}{\beta}+\frac{Max'_t}{\beta}\leq \frac{3Max'_t}{\beta}.\]
    From the assumption that $\beta>3$,  $L_t(z)< Max'_t$. Since the maximum load does not increase
  at time $t$.  Then, we have $L_t(u)> \frac{Max_t}{\beta}$.

  \textbf{Sub case 2:} Node $u$ calls {\splitnbr} steps. This sub case
  occurs when $L'_t(z)>\frac{2Max'_t}{\beta}$. The load of $z$ is shared to
  $u$ to balance their loads, i.e., 
  \[ L_t(u)\geq \frac{\frac{2Max'_t}{\beta}+L'_t(u)}{2}> \frac{Max'_t}{\beta}.\]
  In this sub case, there are two nodes involved, $u$ and $z$.  Note that the only node in
  the system that its load increases at time $t$ is $u$.  We know that
  $L'_t(u)\leq\frac{Max'_t}{\beta}<\frac{Max'_t}{2}$. Then, $u$'s load
  after  load-balancing steps cannot over $Max'_t$ because its load
  less than $\frac{Max'_t}{2}$ and its received load cannot over
  $\frac{Max'_t}{2}$.  Since the maximum load does not increase
  at time $t$.  Therefore, we have $L_t(u)> \frac{Max_t}{\beta}$.
  
  Consider the load of the minimum-loaded node after deletion occurring on it at time
  $t$. When it is more than $\frac{Max'_t}{\beta}$,  load-balancing steps
  are not invoked. From the first case, the minimum load can decrease.
  When it is not over $\frac{Max'_t}{\beta}$,  load-balancing steps
  are invoked. From the second case, the load after load-balancing
  steps is more than $\frac{Max_t}{\beta}$. Thus, the minimum load does not
  decrease after deletion in this case.
\end{proof}

From Lemma~\ref{lem:delete-property}, the minimum load at time $t$ can
decrease in the case that  load-balancing steps are not invoked.
After the minimum load decreases at time $t$, we have
$Min_t>\frac{Max_t}{\beta}$ and the phase changes.  Moreover, the following
condition holds:

At the beginning of each phase at time $t$, the imbalance ratio
guarantee is $\beta$, i.e.,
\[
\frac{Max_t}{Min_t} < \beta.
\]

\subsubsection{Imbalance ratio inside each phase}
\label{sect:inside-phase}

We prove the same invariant in the previous section, i.e.,
\begin{quote} {\sc For any time $t$, the load of any node is not over
    $(\alpha+2)$ times  the minimum load,}
\end{quote}
holds after each operation.

At the beginning of each phase, the imbalance load
guarantee is $\beta$.  At the end, we choose $\beta$ 
such that $\beta\leq\alpha$; this implies that at the beginning of 
each phase, $Max_t\leq\beta Min_t\leq (\alpha+2) Min_t$, as required.

In our analysis later on, we ignore the case of the deletion 
which is not followed with load-balancing steps, because 
the load of the affected can never violate the ratio.  Thus,
the events of this type shall not be considered in our analysis.

To analyze the imbalance ratio, we consider how the load of
each affected node changes. We deal with two easy events first.  
For insertion, we use Lemma~\ref{lem:insert-property} to guarantee 
the load of inserted node and  for deletion, we use 
Lemma~\ref{lem:delete-property} to guarantee the load of deleted node.
For the rest of this section, we are left with the node which effected
from  load-balancing steps.  There are three types of events, i.e.,
 nbr-transfer, nbr-share and min-transfer events.

We summarize the events and how to deal with them as follows.

\begin{itemize}
\item For the {\bf nbr-transfer events}, we focus on a
  lightly-loaded neighbor of the deleted node. It receives an entire load
  of the deleted node.  Lemma~\ref{lem:load-after-splitmax-splitnbr} deals
  with this type of events.
\item For the {\bf nbr-share events}, a lightly-loaded neighbor of the
  deleted node averages out its load with the deleted node. The imbalance
  ratio is also proved in
  Lemma~\ref{lem:load-after-splitmax-splitnbr}.
\item For the {\bf min-transfer events}, we focus on the load of a
  lightly-loaded neighbor of the minimum-loaded node. It receives load
  from the minimum-loaded node.  Lemma~\ref{lem:min-transfer-general}
  deals with this type of events.
\end{itemize}

The next lemma handles the case of  nbr-transfer and nbr-share
events.

\begin{lemma}
  \label{lem:load-after-splitmax-splitnbr}
  Suppose $\alpha\geq 3$ and $\beta\geq
  \frac{3(\alpha+2)}{\alpha}$. Consider the case that a deletion occurs on node
  $u$ at some time $t$ and $u$ calls  load balancing steps. Let $z$
  be the lightly-loaded neighbor of $u$ at time $t$. Then
  $L_t(u)\leq \alpha\cdot Min_t$ and $L_t(z)\leq \alpha\cdot Min_t$.
\end{lemma}

\begin{proof}  
  There are two types of load balancing steps after deletion.

  \textbf{Case 1:} Node $u$ invokes {\splitmax} steps and a
  nbr-transfer event occurs on $z$.  This case occurs when
  $L'_t(z)\leq \frac{2Max'_t}{\beta}$ and $L'_t(u)\leq \frac{Max'_t}{\beta}$. Node $u$
  proceeds by transferring its entire load to $z$, i.e., 
  \[ L_t(z)\leq \frac{2Max'_t}{\beta}+\frac{Max'_t}{\beta}.\] 
  From the invariant, we have $L_t(z)\leq
  \frac{3((\alpha+2)\cdot Min'_t)}{\beta}$.  Since $\beta\geq
  \frac{3(\alpha+2)}{\alpha}$, it follows that $\frac{3(\alpha+2)}{\beta}\leq\alpha$.  Then,
  the load of $z$ at time $t$ is not over $\alpha\cdot Min'_t$. From
  the non-decreasing property of the minimum load, we have
  $L_t(z)\leq\alpha\cdot Min_t$.  
  
  After that, $u$ shares half a load
  of the maximum-loaded node. From the invariant, we have that
  $L_t(u)\leq \frac{(\alpha+2)Min'_t}{2}$. When $\alpha\geq 2$, it follows that
  $\frac{(\alpha+2)}{2}\leq\alpha$.  Then, the load of $u$ after deletion is
  not over $\alpha\cdot Min_t'$ and  we have $L_t(u)\leq
  \alpha\cdot Min_t$ from the non-decreasing property of the minimum load.

  \textbf{Case 2:} Node $u$ invokes {\splitnbr} steps and a
  nbr-share event occurs on $z$. This case occurs when $L'_t(z)>
  \frac{2Max'_t}{\beta}$ and $L'_t(u)\leq \frac{Max'_t}{\beta}$.  From the invariant,
  we have $L'_t(z)\leq Max'_t\leq (\alpha+2)\cdot Min'_t$. Node $u$
  and $z$ share their loads equally, i.e., 
  \[L_t(z)\leq
  \frac{(\alpha+2)\cdot Min'_t+(\alpha+2)\cdot \frac{Min'_t}{\beta}}{2} \leq
  \frac{(\beta+1)(\alpha+2)}{2\beta}\cdot Min'_t.\]
  From the assumption that $\alpha\geq 3$ and $\beta\geq
  \frac{3(\alpha+2)}{\alpha}$, we have that $\beta>\frac{(\alpha+2)}{(\alpha-2)}$.  
  Moreover, when $\beta>\frac{(\alpha+2)}{(\alpha-2)}$ and $\alpha>2$, it follows
  that $\frac{(\beta+1)(\alpha+2)}{2\beta}\leq\alpha$.  Then, the load
  of $z$ and $u$ after load-balancing steps are not over $\alpha\cdot
  Min'_t$. From the non-decreasing property of the minimum load, we have
  $L_t(u)\leq \alpha\cdot Min_t$ and $L_t(z)\leq \alpha\cdot Min_t$.  
\end{proof}

We are left with the min-transfer case. Let $e$ be the min-transfer
event occurring on $z$. Next lemma deals with the case that $e$ is the
first event occurring on $z$ in phase $d$.

\begin{lemma}
  \label{lem:minimum-decrease}
  Consider any phase $d$. Suppose $\beta>2$. If the first event
  occurring on $z$ in that phase is the $i$-th min-transfer event, 
  $L_{t_{i}}(z)\leq (\beta+1)\cdot Min_{t_{i}}.$
\end{lemma}

\begin{proof}
  There are two cases to consider.
  First, we consider the case that $d=1$, i.e., the first phase.
  At the beginning of this phase, the load of any node is equal to $c_0$. 
  Note that the load of $z$ does not change until the first min-transfer event occurs on it.
  Moreover, at time $t_1$, the minimum load is
  also $c_0$ because it cannot decrease and it cannot increase over $c_0$. 
  When the min-transfer event occurs on $z$, it's load increases, i.e.,
  $L_{t_{1}}(z)=c_0+ Min'_{t_1}= 2Min'_{t_{1}}$. From the
  non-decreasing property of the minimum load, we have 
  $L_{t_{1}}(z)\leq 2Min_{t_{1}}$. From the assumption that $\beta>2$, it follows that
  $L_{t_{1}}(z)\leq (\beta+1)\cdot Min_{t_{1}}$.
	  
  For the case that $d > 1$, let $t'$ be the beginning time of phase
  $d$. From the load condition at the beginning of each phase,
   $Max_{t'}<\beta Min_{t'}$. Therefore, $L_{t'}(z)<\beta Min_{t'}$. 
   After the $i$-th min-transfer event occurs on $z$, its load increases,
  i.e., $L_{t_{i}}(z)=L_{t'}(z)+Min'_{t_{i}} \leq \beta\cdot Min_{t'}
  +Min'_{t_{i}}$.  Because the minimum load never decreases in each phase,
  we have  $L_{t_{i}}(z)\leq (\beta+1)\cdot Min_{t_{i}}$ as
  required.
\end{proof}

Now, assume that $e$ is not the first event occurring on
$z$ in phase $d$.

Let $e'$ be the latest event that occurs on $z$ before an event $e$. The
next lemma considers the case where $e'$ is not a min-transfer event,
while Lemma~\ref{lem:genneralizemintransfer2}, which is more involved, considers the case when
$e'$ is a min-transfer event.

\begin{lemma}
  \label{lem:event-before-min-transfer}
  Consider any phase. Suppose $\alpha\geq 3$ and $\beta\geq
  \frac{3(\alpha+2)}{\alpha}$. If any event $e'$ except a min-transfer event occurs on
  node $z$ right before the $i$-th min-transfer event, then after the
  $i$-th min-transfer event occurs on $z$, $ L_{t_{i}}(z)\leq(\alpha +
  1)\cdot Min_{t_{i}}.  $
\end{lemma}

\begin{proof}
  An event $e'$ can be an insertion, or a deletion which is followed with
  load-balancing steps, or  an nbr-transfer event, or an nbr-share event.  
  After event $e'$ occurs on $z$ at time $t'$,
  from Lemma~\ref{lem:insert-property} (for insertions) and
  Lemma~\ref{lem:load-after-splitmax-splitnbr} (for other events), we have
  $L_{t'}(z)\leq\alpha\cdot Min_{t'}$. The load of $z$ increases again
  from  the $i$-th min-transfer event, i.e.,
  $L_{t_{i}}(z)=L_{t'}(z)+Min'_{t_{i}}\leq \alpha\cdot
  Min_{t'}+Min'_{t_{i}}$. Because the minimum load does not
  decrease in each phase, it follows that $L_{t_{i}}(z)\leq (\alpha+1)\cdot
  Min_{t_{i}}$.
\end{proof}

Finally, we consider the case when the latest event $e'$ before event
$e$ is a min-transfer event.  Recall that we categorize the
min-transfer event into two types: the left-transfer event and the
right-transfer event.  The next lemma is a generalization of
Lemma~\ref{lem:genneralizemintransfer} but it deals more with
deletion, {\splitmax} and {\splitnbr}.

\begin{lemma}
  \label{lem:genneralizemintransfer2}
  Suppose $\alpha \geq \frac{3+\sqrt{33}}{2}\approx 4.373$,
  $\beta\geq \frac{3(\alpha+2)}{\alpha}$ and $l\geq 0$. Consider the $i$-th,
  $(i+1)$-th, ..., $(i+l+1)$-th min-transfer events on $z$ in
  any phase. If the $i$-th and $(i+l+1)$-th min-transfer
  events are of the same type, while the $(i+1)$-th, $(i+2)$-th, ...,
  $(i+l)$-th min-transfer events are of type different from that of
  the $i$-th and $(i+l+1)$-th events, then after the $(i+l+1)$-th
  min-transfer event, $L_{t_{i+l+1}}(z)\leq (\alpha+l+1)\cdot
  Min_{t_{i+l+1}}$.
\end{lemma}

\begin{proof}
  Without loss of generality, we assume that the $i$-th and 
  $(i+l+1)$-th min-transfer events occurring on $z$ are 
  right-transfer events; and the $(i+1)$-th, $(i+2)$-th,
  ..., $(i+l)$-th are  left-transfer events. Let $v_{i}$ be the
  minimum-loaded node at time $t_i$.
 
  We first deal with the case that there exists an event except the
  min-transfer event occurring on $z$ between the $i$-th and
   $(i+l+1)$-th min-transfer events. Assume that the latest
  event except the min-transfer event occurs on $z$ right before the
  $k$-th min-transfer event where $i<k\leq i+l+1$. From
  Lemma~\ref{lem:event-before-min-transfer}, after the $k$-th
  min-transfer event, we have
  \[L_{t_k}(z)\leq(\alpha + 1)\cdot Min_{t_k}.\]
  
  After the $k$-th min-transfer event, only the
  min-transfer event can occur on $z$.  Note that after time $t_k$,
  there are at most $i+l-k$ min-transfer events occurring on $z$ at
  time $t_{i+l+1}$. Since, the minimum load never decreases, we have
  that
  
  \[
  L_{t_{i+l+1}}(z) \leq (\alpha+1)\cdot Min_{t_{i+l+1}} + (i+l-k)\cdot
  Min_{t_{i+l+1}},
  \]

  and finally we have $L_{t_{i+l+1}}(z)\leq (\alpha+l+1)\cdot
  Min_{i+l+1}$, because $k> i$.
  
  We are left to consider the case that there is no other events
  occurring on $z$ except the min-transfer event between the $i$-th
  and $(i+l+1)$-th min-transfer events. We follow Claim~\ref{clamin:1} in the proof
  of Lemma~\ref{lem:genneralizemintransfer} using the property that in
  each phase the minimum load does not decrease. Then, we have
  
  \begin{equation}\label{eqn:5}
    L_{t_{i+l+1}}(z) \leq L_{t_{i}}(z) + (l+1)\cdot
    Min'_{t_{i+l+1}}.
  \end{equation}

  Consider the $i$-th min-transfer event on $z$. Recall that, it is a
  right-transfer event. Let $x$ be the node on the right of $v_i$
  right before the $i$-th min-transfer event.  Note
  that, because the $(i+l+1)$-th min-transfer event is also
  a right-transfer event, $x$ must become the minimum-loaded
  node at some point after time $t_i$.  Let $t^*$ be the time that $x$ becomes the minimum-loaded
  node after time $t_i$.
 
  Instead of focusing on node $x$, we consider a set of nodes $P$
  that arranges in a consecutive order from $x$ to $v_{i+l+1}$ between
  time $t_i$ and $t_{i+l+1}$ (see Figure~\ref{fig:2} (b)). We 
  bound the load of $z$ by the load of node in this set. There are two
  cases.

  \textbf{Case 1:} No nodes in set $P$ calls {\minbalance},
  {\splitmax} and {\splitnbr} steps between time $t_i$ and
  $t_{i+l+1}$. In this case, we consider the node $x$. Note that deletion may occur on $x$.  
  There are two sub cases.

  \textbf{Sub case 1.1:} No deletion occurs on $x$ after $t_i$.
  Consider the $i$-th min-transfer event occurring on $z$. The minimum-loaded
  node at time $t_i$ transfers its load to $z$. That means
  $L'_{t_i}(z)\leq L'_{t_i}(x)$. Then, we have
\[
L_{t_i}(z)=L'_{t_i}(z)+Min'_{t_i}\leq L'_{t_i}(x)+Min'_{t_i},
\]

and from Eq.~(\ref{eqn:5}), we have

\begin{equation}\label{eqn:6}
  L_{t_{i+l+1}}(z) \leq L'_{t_i}(x)+Min'_{t_i} + (l+1)\cdot
  Min'_{t_{i+l+1}}.
\end{equation}

After time $t_i$, deletion and load-balancing steps do not occur on
$x$.  The operation that can occur on it is an insertion. Then, $x$'s
load does not decrease.  At time $t^*$, $x$ becomes the minimum-loaded
node. Then, $L'_{t_i}(x)\leq Min_{t^*}$.  From Eq.~(\ref{eqn:6}), we
have

\[L_{t_{i+l+1}}(z) \leq Min_{t^*}+Min'_{t_i} + (l+1)\cdot
Min'_{t_{i+l+1}}.\]

Because the minimum load does not decrease, we have $ L_{t_{i+l+1}}(z)
\leq (l+3)\cdot Min_{t_{i+l+1}}$. When $\alpha\geq 2$, we have
$L_{t_{i+l+1}}(z)\leq (\alpha+l+1)\cdot Min_{t_{i+l+1}}$. Thus, the
lemma holds in this sub case.

\textbf{Sub case 1.2:} Some deletion occurs on $x$ after $t_i$ but
{\splitmax} and {\splitnbr} steps are not invoked. Let $t'$ be the
time that $x$ has the minimum load after deletion occurring on it
between time $t_i$ and $t^*$. From Lemma~\ref{lem:delete-property},
the maximum load at time $t'$ is not over $\beta$ times  the load of
$x$ at time $t'$.  Then, we have
\begin{equation}\label{eqn:max_time_t}
 Max_{t'}\leq\beta L_{t'}(x)\leq\beta L_{t^*}(x). 
 \end{equation}

 Note that the only event that can occur on $z$ between time $t_i$
 and $t'$ is a min-transfer event. Thus, the load of $z$ after $t_i$
 does not decrease. Then, we have $L_{t_i}(z)\leq L_{t'}(z)$. Consider
 Eq.~(\ref{eqn:5}). We have

 \[L_{t_{i+l+1}}(z)\leq L_{t'}(z) + (l+1)\cdot Min'_{t_{i+l+1}}.\]

 From Eq.~(\ref{eqn:max_time_t}), we have 
 \[ L_{t_{i+l+1}}(z)\leq
 Max_{t'} + (l+1)\cdot Min'_{t_{i+l+1}} \leq \beta L_{t^*}(x) +
 (l+1)\cdot Min'_{t_{i+l+1}}. \] 
 Since in each phase, the minimum load never
 decreases, we have $ L_{t_{i+l+1}}(z) \leq (\beta+ l+1)\cdot
 Min_{t_{i+l+1}}$. From the assumption that $\alpha\geq\frac{3+\sqrt{33}}{2}$, it follows that
 $\alpha\geq\beta$ and thus, the lemma holds in this sub case.

 \textbf{Case 2:} At least one node in $P$ calls {\minbalance}
 or {\splitmax} or {\splitnbr} steps between time $t_i$ and
 $t_{i+l+1}$.  Let $y$ be the latest node in this set that calls 
 load-balancing steps at time $\hat{t}$.  We consider the case that
 some deletion occurs on $y$ after $\hat{t}$. We can prove in the same
 way as sub case 1.2 and it follows that $L_{t_{i+l+1}}(z)\leq(\alpha+
 l+1)\cdot Min_{t_{i+l+1}}$.
       
 We are left to consider the case that no deletion occurs on $y$ after
  {\minbalance} or {\splitnbr} or {\splitmax} steps occurring on $y$
 after $\hat{t}$.

 \textbf{Sub case 2.1:} Node $y$ calls {\minbalance} steps at time
 $\hat{t}$. This sub case can be proved in the same way as case 2 in
 Lemma~\ref{lem:genneralizemintransfer}. When $\alpha\geq 1+\sqrt{5}$,
 it follows that $ \frac{L_{t_{i+l+1}}(z)}{Min_{i+l+1}}\leq (\alpha+l+1).$
 Thus, the lemma holds in this sub case.

 \textbf{Sub case 2.2:} Node $y$ calls {\splitnbr} steps at time
 $\hat{t}$. From Lemma~\ref{lem:delete-property}, after {\splitnbr}
 steps, we have $Max_{\hat{t}}\leq \beta L_{\hat{t}}(y)$.  Note that
 the event that occurs on $z$ after $t_i$ is the min-transfer
 event. Then, $z$'s load does not decrease after $t_i$. We have
 $L_{t_i}(z)\leq L_{\hat{t}}(z)$. Consider Eq.~(\ref{eqn:5}). Then, it follows that

 \[L_{t_{i+l+1}}(z) \leq Max_{\hat{t}}+ (l+1)\cdot Min'_{t_{i+l+1}}
 \leq \beta\cdot L_{\hat{t}}(y)+ (l+1)\cdot Min'_{t_{i+l+1}}.\] 

  After time $\hat{t}$, no load-balancing steps or deletion occurs on $y$.  Then,
 after $\hat{t}$, $y$'s load never decreases.  Since the $(i+l+1)$-th 
 min-transfer event is a right-transfer event, $y$ must become the 
 minimum-loaded node at some point.  Let $t^{''}$ be
 the time that $y$ becomes the minimum-loaded node. 
 It follows that
 $L_{\hat{t}}(y)\leq L_{t^{''}}(y)$.  Then, we have
 $L_{t_{i+l+1}}(z)\leq(\beta+ l+1)\cdot Min_{t_{i+l+1}}$. When
 $\alpha\geq\frac{3+\sqrt{33}}{2}$, it follows that $\alpha\geq\beta$ and
 thus, the lemma holds in this sub case.

 \textbf{Sub case 2.3:} Node $y$ calls {\splitmax} steps at time
 $\hat{t}$. After $y$ calls these steps, it transfers its entire load
 to its neighbor $w$. Then, we have $L_{\hat{t}}(w)\geq 2 Min_{\hat{t}}$. After
 that, $y$ is relocated. Thus, we consider node $w$ instead. 
 
 We show that $w$ is in $P$. Assume by contradiction that $w$ is not in $P$.
 Note that the position of $w$ can be left or right of $y$.
 Consider the case that $w$ is right of $y$. $w$ can be outside $P$ when 
 $y$ must be the rightmost node in $P$. Because $v_{i+l+1}$ is the rightmost node 
 in $P$, this case contradicts. Consider the case that $w$ is left of $y$. 
 In this case, $y$ must be the leftmost node in $P$ and $y$ transfers its load to the node
 outside $P$. Note that the node which $y$ transfers its load is $z$.
 This contradicts that no events occur on $z$ except the min-transfer event.
  
 Consider the case that some
 deletion occurs on $w$ after $\hat{t}$. This case can be proved in the same
 way as sub case 1.2.

 Now, we assume  that no deletion occurs on $w$ after
 $\hat{t}$.  Recall that at time $\hat{t}$, we have $2
 Min_{\hat{t}}\leq L_{\hat{t}}(w)$.  Since the $(i+l+1)$-th
 min-transfer event is a right-transfer event, $w$ becomes the
 minimum-loaded node at some point after $\hat{t}$. Let $t^*$ be the time that $w$
 becomes the minimum-loaded node after $\hat{t}$.  After time
 $\hat{t}$, load-balancing steps and deletion do not occur on $w$.  Then,
  $w$'s load does not decrease.  We have
 $L_{\hat{t}}(w)\leq L_{t^*}(w)$.  Moreover, we have
 $L_{\hat{t}}(w)\leq Min'_{t_{i+l+1}}$ from the non-decreasing property
 of the minimum load.  Thus, $2 Min_{\hat{t}}\leq Min'_{t_{i+l+1}}$.

 Consider Eq.~(\ref{eqn:5}). We divide it by $Min_{i+l+1}$,
 i.e.,\begin{eqnarray*}\label{z-and-m-general}
   \frac{L_{t_{i+l+1}}(z)}{Min_{i+l+1}} &\leq& \frac{L_{t_i}(z)+
     (l+1)\cdot Min'_{t_{i+l+1}}}{Min_{i+l+1}}\\
   &\leq& \frac{L_{t_i}(z)+ (l+1)\cdot
     Min'_{t_{i+l+1}}}{Min'_{t_{i+l+1}}}\\
   & = & (l+1)+\frac{L_{t_i}(z)}{Min'_{t_{i+l+1}}}.
 \end{eqnarray*}  

 From the invariant, it follows that $L_{t_i}(z)\leq (\alpha+2)\cdot
 Min_{t_i}\leq (\alpha+2)\cdot Min_{\hat{t}}$ when
 $t_i<\hat{t}$. Recall that $2 Min_{\hat{t}}\leq Min'_{t_{i+l+1}}$.
 Then,
\[
\frac{L_{t_{i+l+1}}(z)}{Min_{i+l+1}} \leq (l+1)+\frac{(\alpha+2)\cdot
  Min_{\hat{t}}}{2Min_{\hat{t}}}.
\]

When $\alpha>2$, it follows that $\frac{(\alpha+2)}{2}\leq\alpha$. Then, we
have $ \frac{L_{t_{i+l+1}}(z)}{Min_{i+l+1}}\leq (\alpha+l+1).  $ Thus,
the lemma holds in this sub case.
\end{proof}

From Lemma~\ref{lem:minimum-decrease},
~\ref{lem:event-before-min-transfer}
and~\ref{lem:genneralizemintransfer2}, we conclude the effect of the
min-transfer event on any node $z$. We omit the proof because it is
similar to the proof of Lemma~\ref{lem:min-transfer}.

\begin{lemma}
  \label{lem:min-transfer-general}
  Consider any phase $d$.  Suppose $\alpha \geq
  \frac{3+\sqrt{33}}{2}\approx 4.373$ and $\beta\geq
  \frac{3(\alpha+2)}{\alpha}$. After the $i$-th min-transfer event occurs on
  $z$, $ L_{t_i}(z)\leq(\alpha + 2)\cdot Min_{t_i}.  $
\end{lemma}

Finally, we are ready to prove the imbalance ratio guarantee. In each
phase, the load of any node can be changed by insertion, deletion,
nbr-transfer, nbr-share and min-transfer events.  From
Lemma~\ref{lem:insert-property}, \ref{lem:delete-property},
\ref{lem:load-after-splitmax-splitnbr}
and~\ref{lem:min-transfer-general},  we can conclude the upper bound of
load of any node in any phase after these events.

\begin{theorem}\label{thm:maxload2}
  Consider any phase. Suppose $\alpha \geq
  \frac{3+\sqrt{33}}{2}\approx 4.373$ and $\beta\geq
  \frac{3(\alpha+2)}{\alpha}$. For any node $u \in V$, after any event at time
  $t$, $ L_t(u)\leq(\alpha+2)\cdot Min_t. $ Thus, the imbalance ratio
  of the algorithm in the general case is a constant.
\end{theorem}	

\subsection{Cost of the algorithm in the general case}
\label{sect:imbalanceratioproof2}

We analyze the cost of our algorithm in the general case. Recall that, the
cost of moving a key from node $u$ to another node $v$ is counted as a unit
cost. In our algorithm, there are four operations, i.e., insertion,
deletion, {\minbalance} and {\split}. Again, we prove the amortized
cost by potential method and use the same potential function of
Ganesan {\em et al.}.

\begin{theorem}\label{thm:amortized_cost} 
 Suppose $\alpha>2(1+\sqrt{3})\approx 5.464$.  The amortized costs of our algorithm are constant.
\end{theorem}

\begin{proof}
  We use the  potential function in~\cite{GanesanBGM04-vldb}, i.e., $\Phi=
  \frac{c(\sum^{n}_{i=1}(L_t(N_i))^2)}{\bar{L}_t}$, where $\bar{L}_t$ is the
  average load at time $t$ and $N_i$ be the $i$-th node that manages
  a range $[R_{i-1}, R_i)$. We show that the gain in potential
  when insertion or deletion occurs is at most a constant and the drop
  in potential when {\minbalance} or {\split} operation occurs pays
  for the cost of the operation. These imply that the amortized costs
  of insertion and deletion are constant.

  We prove the cost of insertion and {\minbalance} operation in the
  same way as Theorem~\ref{thm:amortized_cost_insertonly}. We are left
  to consider a deletion and {\split} operation.

{\bf Deletion}: When a deletion occurs on $N_j$ at time $t$, the
gain in potential is at most

\begin{eqnarray*}
  \Delta\Phi & \leq &c\frac{\left(\sum_{i\in
        V}{L_t(N_i)^2}\right)}{\left(\bar{L}_t-
      \frac{1}{n}\right)}-c\frac{\left(\sum_{i\in
        V}{L_t(N_i)^2}\right)}{\bar{L}_t}\\ 
  &=&c\frac{\left(\sum_{i\in V}{L_t(N_i)^2}\right)(\frac{1}{n})}{\bar{L}_t\cdot\left(\bar{L}_t-\frac{1}{n}\right)}.
  \end{eqnarray*}

  We know that $\bar{L}_t\geq Min_t$. From the invariant, we have that
  $L_t(u)\leq(\alpha+2)\cdot Min_t\leq(\alpha+2)\cdot\bar{L}_t$ for any node $u$.  Using the
  fact that $\bar{L}_t-\frac{1}{n} \geq \frac{\bar{L}_t}{2}$ where $n \geq 2$ and
  $\bar{L}_t \geq 1$, we have

  \begin{eqnarray*}
\Delta\Phi& \leq&
c\frac{\left((\alpha+2)\cdot\bar{L}_t\right)^2}{\bar{L}_t\cdot\left(\frac{\bar{L}_t}{2}\right)}
\leq 2c(\alpha+2)^2.
  \end{eqnarray*}

  Hence, the amortized cost of deletion is a constant because the actual cost
  of deletion is a unit cost.

  {\bf Split}: Consider the load of node $N_j$, whose calls {\split}
  at time $t$. When its load is less than $\frac{Max_t}{\beta}$, 
  load-balancing steps are called.  There are two cases, i.e.,
  {\splitnbr} and {\splitmax} steps.

  {\bf Case {\splitnbr}:} Let $N_k$ be the lightly-loaded neighbor of
  $N_j$ when $N_j$ calls {\splitnbr} steps. Node $N_k$ moves its keys
  to $N_j$ to balance their loads.  The drop in potential is
  \begin{eqnarray*}
    \Delta\Phi&=&
    \frac{c\left(L_t(N_j)^2+L_t(N_k)^2-2\left(\frac{L_t(N_j)+L_t(N_k)}{2}\right)^2\right)}{\bar{L}_t}\\
    &=&c\frac{\left(\frac{L_t(N_j)^2}{2}+\frac{L_t(N_k)^2}{2}-\left(L_t(N_j)L_t(N_k)\right)\right)}{\bar{L}_t}\\
    &=&c\frac{\left(L_t(N_k)-L_t(N_j)\right)^2}{2\bar{L}_t}\\		
    &=&c\frac{\left(L_t(N_k)-L_t(N_j)\right)}{2}\frac{\left(L_t(N_k)-L_t(N_j)\right)}{\bar{L}_t}.\\
  \end{eqnarray*}  
  
  This case occurs when
  $L_t(N_k)>\frac{2Max_t}{\beta}$ and $L_t(N_j)<\frac{Max_t}{\beta}$. Then,
  \begin{eqnarray*}
    \Delta\Phi&\geq
    &c\frac{\left(L_t(N_k)-L_t(N_j)\right)}{2}\frac{\left(\frac{2Max_t}{\beta}-\frac{Max_t}{\beta}\right)}{\bar{L}_t}\\
    &=&c\frac{\left(L_t(N_k)-L_t(N_j)\right)}{2}\frac{\left(\frac{Max_t}{\beta}\right)}{\bar{L}_t}\\
    &\geq
    &c\frac{\left(L_t(N_k)-L_t(N_j)\right)}{2}\left(\frac{1}{\beta}\right).
  \end{eqnarray*}

  The number of moved keys when {\splitnbr} steps are invoked is
  $\frac{\left(L_t(N_k)-L_t(N_j)\right)}{2}$. For any $c>\beta$, we have
  $\Delta\Phi\geq \frac{\left(L_t(N_k)-L_t(N_j)\right)}{2}$.  Thus,
  the drop in potential pays for the data movement.

  {\bf Case {\splitmax}:} Node $N_j$ transfers its entire load to its
  adjacent node $N_k$ and then shares half the load of the
  maximum-loaded node $N_l$.  The drop in potential is

  \begin{eqnarray*}
    \Delta\Phi&=& c\frac{\left(L_t(N_j)^2+L_t(N_k)^2+L_t(N_l)^2
        -2\left(\frac{L_t(N_l)}{2}\right)^2-\left(L_t(N_k)+L_t(N_j)\right)^2\right)}{\bar{L}_t}\\
    &=&
    c\frac{\left(\frac{L_t(N_l)^2}{2}-2L_t(N_k)L_t(N_j)\right)}{\bar{L}_t}\\
    &=& \frac{c}{\bar{L}_t}
    L_t(N_l)^2\left(\frac{1}{2}-\frac{2L_t(N_k)L_t(N_j)}{L_t(N_l)^2}\right).
\end{eqnarray*}  

From the algorithm, we know that $L_t(N_k)\leq \frac{2Max_t}{\beta}$
and $L_t(N_j)\leq \frac{Max_t}{\beta}$. Then,
\begin{eqnarray*}  			
  \Delta\Phi&\geq& c L_t(N_l)\left(\frac{1}{2}-\frac{2\frac{2
        Max_t}{\beta} \cdot\frac{ Max_t}{\beta}}{L_t(N_l)^2}\right)\\
  &\geq& c\cdot L_t(N_l)\left(\frac{1}{2}-\frac{4}{\beta^2
    }\right).
\end{eqnarray*}

When {\splitmax} steps are invoked, the number of moved keys is
$\frac{L_t(N_l)}{2}+L_t(N_j)\leq L_t(N_l)$. For any
$c>\left(\frac{2\beta^2}{\beta^2-8}\right)$, it follows that $\Delta\Phi\geq
L_t(N_l)$ and thus, the drop in potential pays for the key movement.

Thus, the amortized costs of the algorithm are constant.
\end{proof}

\section{The cost in real networks}
\label{sect:p2p}

In this section, we discuss how Ganesan {\em et
  al.}~\cite{GanesanBGM04-vldb} dealt with  global information and,
again, discuss the comparison between this line of work, which this
paper extends, and the work of Karger and
Ruhl~\cite{Karger03newalgorithms, Karger04simpleefficient} when
considering real networks.

In the P2P networks, there is no centralized server to
provide the information.  If any node requires information about another
node, it must send  messages to that node.  Besides the data
movement cost, there is another cost to be considered, the
communication cost.  The communication cost is defined as a number
of messages that required for complete the operation.

We shall discuss how Ganesan {\em et al.} implement the idea on real
networks.  Their implementation is based on skip
graphs~\cite{Aspnes-SODA03}.  Skip graphs support find operation,
node insertion, and node deletion with $O(\log n)$ messages with high
probability.  Also adjacent nodes can be contact with $O(1)$ messages.
Ganesan {\em at al.}  use two skip graphs: one where nodes are
ordered by their minimum key in their ranges; another where nodes are
ordered by their loads.  Therefore, global information can be found
with $O(\log n)$ messages, and each partition change costs at most
$O(\log n)$ messages.  We note that while this cost is more than the
constant cost of data movement, usually $O(\log n)$ messages are
required for finding the node for each key, and thus these cost can be
amortized with the searching cost.

As in Ganesan {\em et al.} works, Karger and
Ruhl~\cite{Karger03newalgorithms, Karger04simpleefficient} simplify
the cost model by not considering how one could find a given node in
the system.  Therefore, unless they maintain a global directory of
nodes, using known data structures for p2p systems, they still need
$O(\log n)$ messages as well.

\bibliography{reorder}
\end{document}